\documentclass{amsart}
\usepackage{amssymb}
\usepackage{verbatim}
\newtheorem{theorem}{Theorem}[section]

\newtheorem{proposition}[theorem]{Proposition}

\newtheorem{corollary}[theorem]{Corollary}

\theoremstyle{definition}
\newtheorem{definition}[theorem]{Definition}
\newtheorem{example}[theorem]{Example}

\theoremstyle{remark}
\newtheorem{remark}[theorem]{Remark}
\numberwithin{equation}{section}

\newcommand{\cE}{{\mathcal E}}

\newcommand{\M}{\mathcal{M}}
\newcommand{\A}{{\mathcal A}}

\newcommand{\cM}{{\mathcal M}}
\newcommand{\cN}{{\mathcal N}}

\newcommand{\X}{{\mathcal X}}

\newcommand{\Rn}{{\rm I\!R}}

\newcommand{\cS}{{\mathcal S}}
\newcommand{\I}{1\!{\mathrm l}}

\newcommand{\tM}{\widetilde{\mathcal{M}}}

\begin{document}

\title{On applications of Orlicz Spaces to Statistical Physics\\}

\author{W. A. Majewski}

\address{Institute of Theoretical Physics and Astrophysics, The Gdansk University, Wita Stwosza 57,\\
Gdansk, 80-952, Poland\\
E-mail: fizwam@univ.gda.pl}

\author{L. E. Labuschagne}

\address{Internal Box 209, School of Comp., Stat., $\&$ Math. Sci.\\
NWU, PVT. BAG X6001, 2520 Potchefstroom, South Africa\\
E-mail: Louis.Labuschagne@nwu.ac.za}

\begin{abstract}
We present a new rigorous approach based on Orlicz spaces for the description of the statistics of large regular statistical systems, both classical and quantum. This approach has the advantage that statistical mechanics is much better settled. In particular, a new kind of renormalization leading to states having a well defined entropy function is presented.
\end{abstract}

\maketitle



\section{Introduction}

To indicate reasons why (classical as well as non-commutative) Orlicz spaces are emerging in the theory of (classical and quantum) Physics we begin with a simple question asking when a physicist knows that a certain quantity is an observable. Obviously, one answers - an observable is known when also a function of this observable is  known. A nice illustration of this way of thinking is provided by Classical Mechanics - for example: knowing a coordinate one knows also a potential (being a function of coordinates), etc. 
It is worth pointing out that exactly this feature of observables was probably a motivation for Newton to develop calculus and to use it in his laws of motion. 

On the other hand, in Statistical Physics the same question seems to be more subtle. Namely, let $(X, \Sigma, m)$ be a probability space and $u$ an observable and thus a random variable (so a measurable function). The question just posed implies that we wish to at least know the average of $\langle u\rangle_m$
as well as $\langle F(u)\rangle_m$ for a large class of functions $F$. Assume that $F$ has the Taylor expansion $F(x) = \sum_i c_i x^i $. Our demands mean that $\langle F(u)\rangle_m = \sum_i c_i \langle u^i\rangle _m$ should be well defined. However this implies that one should be able to select a subset of observables, say ``regular'' observables, for which all moments are finite.
It is worth pointing out that for the special case of Dirac measures (so for point masses) the answer given by Classical Mechanics can be reproduced.

Let us consider the question posed above in detail in the context of probability theory. Denote by
 ${\cS}_{m}$  the set of the densities of all the probability measures equivalent to $m$, i.e.,

$${\cS}_{m} = \{ f  \in L^1(m): f>0 \quad m-a.s., E(f) =1 \}.$$

Here $E(f)\equiv \langle f\rangle_m$ stands for $\int f(x) dm(x)$. 
${\cS}_{m}$ can be considered as a set of (classical) states and its natural ``geometry''  comes from embedding ${\cS}_{m}$ into $L^1(m)$. However, it is worth pointing out that the 
Liouville space technique demands $L^2(m)$-space, whilst the 
employment of interpolation techniques needs other $L^p$-spaces with $p \geq 1$.

Turning to the moment problem, let us consider a class of 
moment generating functions; so fix $f \in {\cS}_{m}$ and take a real random variable $u$ on $(X, \Sigma, f dm)$.
Define (see \cite{B})

$$ \hat{u}_f(t) = \int exp(tu) f dm, \qquad t \in \Rn $$

and denote by $L_f$  the set of all random variables such that
 $\hat{u}_f$ is well defined in a neighborhood of the origin $0$, and
 the expectation of $u$ is zero.

One can observe that in this way a nice selection of (classical) observables was made (\cite{B}, and/or \cite{PS}) in that
 all the moments of every $u \in L_f$ exist and they are the values at $0$ of the derivatives of $\hat{u}_f$.

 But it is important to note that $L_f$ is actually a closed subspace of the {\bf Orlicz space} based on the exponentially growing function $\cosh - 1$ (see  \cite{PS}).
Consequently, one may say that even in classical statistical Physics one could not restrict oneself to merely $L^1(m)$, $L^2(m)$, $L^{\infty}(m)$
and the interpolating $L^p(m)$ spaces. 
Another argument in favour of Orlicz spaces, was provided by Cheng and Kozak \cite{CK}. Namely it seems that the natural framework within which certain non-linear integral equations of Statistical Mechanics can be studied, is provided by Orlicz spaces. In particular, the Orlicz space defined by the Young function $u \mapsto e^{|u|} - |u| - 1$ was playing a distinguished role (see \cite{CK} for details).
In other words, {\bf  generalizations of $L^p$-spaces - Orlicz spaces - do appear.}

But there is a second important problem. Statistical Physics aims at explaining thermodynamics. To this end one should have well defined so called state-functions. A nice example of such function is given by the entropy function which has an exceptional status among other state functions. Entropy is defined:

\begin{enumerate}
\item $H(f) = - \int f(x)\log f(x) d\mu$, $f \in {\cS}_{\mu}$, for the classical (continuous) case;
\item $ S(\varrho) = - \mathrm{Tr} \varrho \log \varrho $, $\varrho$ a density matrix, for the quantum case.
\end{enumerate}

The problem is that both definitions can lead to divergences. To illustrate the seriousness of this problem, we firstly consider the quantum case where we will follow Wehrl \cite{Weh} and Streater \cite{S1,S2}; see also \cite{LM1} . Let $\varrho_0$
be a quantum state (a density matrix) and $S(\varrho_0)$ its von Neumann entropy. Assume $S(\varrho_0)$
to be finite. It is an easy observation that in any neighborhood of $\varrho_0$ (given by the trace norm, so in the 
sense of quantum $L^1$-space) there are plenty of states with infinite entropy. 
One can say more (see Wehrl \cite{Weh} p. 241); the set of ``good'' density matrices $\{ \varrho : S(\varrho) < \infty \}$ is merely a meager set.
This should be considered alongside the thermodynamical rule which tells us that entropy should be a state function which is increasing in time. Thus we run into serious problems with explaining the phenomenon of return to equilibrium and with the second law of thermodynamics.

Turning to classical continuous entropy, we mention only that for $f \in L^1$ the functional $H(f)$ is not well defined - see \cite{Bour}, Chapter IV, \S 6, Exercise 18. (For other arguments see Section 4.)

Attempting to find a solution to the problems outlined above, we propose to replace the pair of Banach spaces  
\begin{equation} 
 \langle L^{\infty}(X, \Sigma, m), L^1(X, \Sigma, m)\rangle
 \end{equation}
appearing in standard approaches to statistics and Statistical Physics,
with the pair of Orlicz spaces (or pairs equivalent to this one, see the ensuing Sections)
\begin{equation}
\label{KwantOrlicz}
 \langle L^{\cosh -1}, \ L\log(L+1)\rangle.
 \end{equation}

The first Orlicz space  $L^{\cosh -1}$ appears as the proper framework for describing the set of regular observables (cf arguments given prior to the discussion of the second problem). The second Orlicz space $L\log(L+1)$, is the space defined by the Young function $x \mapsto x\log(x+1), \ x\geq 0$. This space is nothing but an equivalent renorming of the K{\"o}the dual
of $L^{\cosh -1}$, as
 $\cosh(x) -1$ and $x \mapsto x\log(x + \sqrt{1+x^2}) - \sqrt{1 + x^2} +1, \ x\geq 0$
are complementary Young's functions, with the latter function equivalent to $x \mapsto x\log(1 +x), \ x\geq 0$ (see the next Section for details). 

To appreciate the significance of our choice (for details see the ensuing Sections) we note that models considered in Statistical Physics and Quantum Field Theory  are par excellence large systems, i.e. systems with infinite degrees of freedom. \textit{Our new approach is designed
 exactly for large systems}. Furthermore, 
we note that the condition 
$f \in L\log(L+1)$ guarantees the finiteness of classical (continuous) and quantum entropies (for finite measure case) as well as legitimising the  
consideration of elements of $L\log(L+1)$ as continuous functionals over the set of regular observables of a system. Thus $L\log(L+1)$ is home to the states of a regular statistical system. 
On the other hand, it should be mentioned that in both the classical and quantum cases, an analysis of problems such as return to equilibrium and entropy production, demands more general Banach spaces than $L^p$-spaces - see
\cite{Vil} and the references given therein for the classical case, and \cite{LM1} for the quantum case.

Consequently, we propose {\bf a new rigorous approach} for the description of statistics of large regular statistical systems
{\bf having the advantage that statistics is better settled.} This was obtained by means of the ``regularization'' of admissible states. Namely, $L^{\cosh - 1}$ can be seen as an enlarged family of observables.
Consequently, the (K\"othe) dual  space consist of more regular states (cf Theorem \ref{2.6}). This is a new way
of removing ``non-physical'' states which lead to infinities. Thus a kind of renormalization is proposed.

The next important point to note here, is the fact that Quantum Theory is by nature probabilistic. Therefore the proposed new approach is especially important for a description of quantum systems, and the presented quantization of 
classical regular systems is essential. The presented quantization of classical regular systems reveals rather strikingly the difference between systems associated with factors of type III and II, i.e. large systems, and those which are associated with type I factors (see the end of Section 6). Namely, for the von Neumann algebra $B(H)$ (type I factor) the quantization based on the Dodds, Dodds, de Pagter approach (see Section 5) forces an employment of Banach function spaces which are based on completely atomic measure space, with all atoms having equal measure. This implies that for simple systems (with finite degree of freedom) the standard pair of algebras $\left\langle B(H), L^1(B(H))\right\rangle$
is unchanged ($L^1(B(H))$ stands for the trace class operators). In other words \textit{the regularization procedure which we propose is effective for large systems.}   

The paper will be organized as follows: in Section 2 we review some of the standard facts on (classical) Orlicz spaces. Then classical regular systems are described (Section 3). In Section 4 we indicate how our approach can be extended to the infinite measure case. In particular, certain questions around the Boltzmann equation are considered.
In Section 5 we provide a brief account on non-commutative Orlicz spaces (which is taken from \cite{LM}). Section 6 is devoted to the study of regular non-commutative statistical systems. In particular, the quantization of the Orlicz space approach to regular systems is presented.

\section{Classical Orlicz spaces }

Let us begin with some preliminaries (for details we refer to \cite{KS}, \cite{BS}, and \cite{RR}).

\begin{definition}\cite{BS}
Let $\psi: [0, \infty) \to [0, \infty]$ be an increasing and left-continuous function such that $\psi(0)=0$. Suppose that on $(0,\infty)$ 
$\psi$ is neither identically zero nor identically infinite. Then the function $\Psi$ defined by
\begin{equation}
\label{psi}
\Psi(s) = \int_0^s \psi(u) du, \qquad (s\geq 0)
\end{equation}
is said to be a Young's function.
\end{definition}

Clearly, $x \mapsto \cosh(x) -1$, $x \mapsto x\log(x + \sqrt{1+x^2}) - \sqrt{1 + x^2} +1$, $x \mapsto x \log(x+1)$ are Young's functions while $x \mapsto x\log x$ is not.
\begin{definition}\cite{BS}, \cite{RR}
\label{delta}
\begin{enumerate}
\item A Young's function $\Psi$ is said to satisfy the $\Delta_2$-condition if there	 exist $s_0 > 0$ and $c>0$ such that 
\begin{equation}
\Psi(2s) \leq c \Psi(s) < \infty, \qquad (s_0 \leq s < \infty).
\end{equation}
If $\Psi$ satisfies the above condition for $s_0=0$, we say that it satisfies the $\Delta_2$-condition \emph{globally}.
\item  A Young's function $\Phi$ is said to satisfy $\nabla_2$-condition if there exist $x_0 >0$ and $l >1$ such that
\begin{equation}
\Phi(x) \leq \frac{1}{2l} \Phi(lx)
\end{equation}
for $x \geq x_0$.
\end{enumerate}
If $\Phi$ satisfies the above condition for $x_0=0$, we say that it satisfies the $\nabla_2$-condition \emph{globally}.
\end{definition}
It is easy to verify that the Young's functions, given prior to Definition \ref{delta}, $x \mapsto x\log(x + \sqrt{1+x^2}) - \sqrt{1 + x^2} +1$, $x \mapsto x\log(x+1)$ ($x \mapsto \cosh(x) -1$),  satisfy the $\Delta_2$-condition ($\nabla_2$-condition, respectively).
 
We also need
\begin{definition}\cite{BS}
 Let $\Psi$ be a Young's function, represented as in (\ref{psi}) as the integral of $\psi$. Let
 \begin{equation}
 \label{1}
 \phi(v) = \inf \{ w: \psi(w)\geq v \}, \qquad (0\leq v \leq \infty).
 \end{equation}
 Then the function
 \begin{equation}
 \label{2}
 \Phi(t) = \int_0^t \psi(v) dv, \qquad (0\leq t \leq \infty)
 \end{equation}
 is called the complementary Young's function of $\Psi$.
 \end{definition}

We note that if the function $\psi(w)$ is continuous and monotonically increasing, then $\phi(v)$ is a function exactly inverse to $\psi(w)$. Consequently, as
\begin{equation}
\cosh(x) - 1 = \int_0^x \sinh(v) dv,
\end{equation}
and $sinh(x)$ has a well defined  inverse: $\mathrm{arcsinh}(x)$, we arrive at the second Young's function, namely
\begin{equation}
x\log(x + \sqrt{1+x^2}) - \sqrt{1 + x^2} +1 = \int_0^x \mathrm{arcsinh}(v)dv.
\end{equation}
We have

\begin{corollary}
$x\log(x + \sqrt{1+x^2}) - \sqrt{1 + x^2} +1$  and $\cosh x -1$
are complementary Young's functions.
\end{corollary}

Let $L^0$ be the space of measurable 
functions on some $\sigma$-finite measure space $(Y, \Sigma, \mu)$. Orlicz spaces are defined in:

\begin{definition}
The Orlicz space 
$L^{\Psi}$ associated with  $\Psi$ is defined to be the set 
\begin{equation}\label{3}
L^{\Psi} \equiv L^{\Psi}(Y, \Sigma, \mu) = \{f \in 
L^0 : \Psi(\lambda |f|) \in L^1 \quad \mbox{for some} \quad \lambda = \lambda(f) > 0\}.
\end{equation}
\end{definition}

This space turns out to be a linear subspace of $L^0$, and $L^{\Psi}$ becomes a Banach space when 
equipped with the so-called Luxemburg-Nakano norm 
$$\|f\|_\Psi = \inf\{\lambda > 0 : \|\Psi(|f|/\lambda)\|_1 \leq 1\}.$$
Here $\| \cdot \|_1$ stands for $L^1$-norm.
An equivalent - Orlicz norm, for a pair $(\Psi, \Phi)$ of complementary Young's functions is given by
$$\|f\|_\Phi = \sup\{ \int|fg|d\mu : \int\Psi(|g|) d\mu \leq 1 \}.$$
If $\Psi$ satisfies the $\Delta_2$ condition globally, $L^\Psi$ is more regular in the sense that then 
$L^{\Psi}(Y, \Sigma, \mu) = \{f \in L^0 : \Psi(|f|) \in L^1 \}$. In the case of finite measures, $\Psi$ only needs to satisfy $\Delta_2$ for large values of $t$ for this equality to hold. (See \cite[Theorem III.1.2]{RR} 

Clearly, (classical) $L^p$-spaces are nice examples of Orlicz spaces. 
Other useful examples, so called Zygmund spaces, are defined as follows (cf \cite{BS}):

\begin{itemize}
\item $L\log L$ is defined by the following Young's function
$$s\log^+ s = \int_0^s \phi(u) du$$
where $\phi(u) =0 $ for $0\leq u\leq1$ and $\phi(u) = 1 + \log  u$ for $1< \infty $, where $\log^+x = max (\log x, 0)$.
Note that this Young's function is 0-valued on all of $[0,1]$, and not just for $s=0$.
\item $L_{exp}$ is defined by the Young's function
$$ \Psi(s) = \int_0^s \psi(u)du, $$
where $\psi(0) = 0$ , $\psi(u) = 1$ for $0<u<1$, and $\psi(u)$ is equal to $e^{u -1}$ for $1 < u < \infty$.
Thus
$\Psi(s) = s$ for $0\leq s \leq 1$ and  $\Psi(s) = e^{s - 1}$ for $1< s < \infty$.
\end{itemize}

To understand the role of Zygmund spaces the following result will be helpful, see \cite{BS}:
\begin{theorem}
\label{2.6}
Let $Y, \Sigma, m)$ be a finite measure space with $m(Y) = 1$. The continuous embeddings
\begin{equation}
L^{\infty} \hookrightarrow L_{exp} \hookrightarrow L^p \hookrightarrow L\log L \hookrightarrow L^1
\end{equation}
hold for all p satisfying $1<p< \infty$. Moreover, $L_{exp}$ may be identified with the Banach space dual of $L\log L$.
\end{theorem}

More generally, for  a pair $(\Psi, \Phi)$ of complementary Young's functions with the function $\Psi$ satisfying $\Delta_2$-condition and the function $\Phi(s) = 0$ if and only if $s=0$, one has that $(L^{\Psi})^* = L^{\Phi}$ (cf \cite{RR}).

Theorem \ref{2.6} is a particular case of the following fact (for all details see \cite{BS}): since any classical Orlicz space $X$ is a rearrangement-invariant Banach function space (over a resonant measure space), one has 
\begin{equation}
\label{ha}
L^1\cap L^{\infty} \hookrightarrow X \hookrightarrow L^1 + L^{\infty}
\end{equation}

For the finite measure case (\ref{ha}) is simplified. Namely, one has

\begin{equation}
L^{\infty} \hookrightarrow X \hookrightarrow L^1
\end{equation}

We note that $L^1 \cap L^{\infty}$ is therefore the smallest Orlicz space while $L^1 + L^{\infty}$ is the largest one.

Finally, we will write $F_1 \succ F_2$ if and only if $F_1(bx) \geq F_2(x)$ for $x\geq 0$ and some $b>0$, and we say that the functions $F_1$ and $F_2$ are equivalent, $F_1 \approx
F_2$, if $F_1\prec F_2$ and $F_1\succ F_2$.

\begin{example}
\label{2.7}
Consider, for $x>0$
\begin{itemize}
\item $F_1(x) = x\log(x + \sqrt{1+x^2}) - \sqrt{1 + x^2} +1 = \int_0^x \log(s + \sqrt{1+s^2}) ds$,
\item $F_2 = k x\log x = k \int_0^x (\log s + 1)ds$, $k > e$.
\end{itemize}
Then $F_1 \succ F_2$.
\end{example}

\begin{remark}
\label{uwaga1}
\begin{enumerate}
\item Recall that $x \mapsto x\log x$ is not a Young's function. Therefore it does not make sense to speak of the Orlicz space $L^{x\log x}$.
\item If $\Psi \succ F$, $\Psi$ is a Young's function satisfying $\Delta_2$-condition, and the function $F$ is bounded  below by $- c$, then for $f \in L^{\Psi}$
the integral $\int F(f)(u) dm(u)$ is finite provided that the measure $m$ is finite.
\end{enumerate}
\end{remark}

To see \ref{uwaga1} (2)
we note: by the definition of Orlicz spaces $f \in L^{\Psi}$ implies $\int \Psi(\lambda |f|)(u) dm(u) < \infty$ for some $\lambda$.
Further, as $x \mapsto \Psi(x)$ satisfies the $\Delta_2$-condition, the set $\{ f \in L^0; \int \Psi(|f|) dm(u) < \infty\}$ is a linear space. Therefore, $\lambda > 0$ can be taken arbitrarily.
Hence
$$\infty > \int \Psi(\lambda|f|) dm \geq \int F(|f|)dm \geq -c \cdot m(\cE_f)$$
for a proper choice of $\lambda$ (for example: $\lambda = b =k > e $), where $\cE_f = \{ u: F(|f|)(u) <0 \}$.
Finally,
$\int F(|f|)dm$ is finite if and only if $\int c^{\prime} F(|f|)dm$ is finite, where $c^{\prime}$ is an arbitrary fixed positive number. Thus, we arrived at

\begin{corollary}
\label{dwa1}
Let $(X, \Sigma, m)$ be a probability space.
Putting $\Psi(x) = x\log(x + \sqrt{1+x^2}) - \sqrt{1 + x^2} +1 $ and $F(x) = k x \log x $ where $k >e$ is a fixed positive number we obtain: $H(f)$ is finite provided that $f \in L_+^{\Psi}.$
\end{corollary}

The equivalence relation $\approx$ on the set of Young's functions defined prior to Example \ref{2.7} leads to classes of Young's functions. The principal significance of this concept follows from:

\begin{theorem}(\cite{RR})
Let $\Phi_i$, $i =1,2$ be a pair of equivalent Young's functions. Then $L^{\Phi_1} = L^{\Phi_2}$.
\end{theorem}

Consequently, a pair of Orlicz spaces $(X, X^{\prime})$ where $X^{\prime}$ stands for the (K\"othe) dual of $X$ can be determined using different but equivalent pairs of complementary Young's functions. We will use this strategy to replace the Orlicz space defined by $x \mapsto x\log(x + \sqrt{1+x^2}) - \sqrt{1 + x^2} +1 $ by the Orlicz space $L \log(L+1)$ and to legitimize the pair $\langle L^{\cosh - 1}, L\log(L+1)\rangle$. For the finite measure case we will see that one even replace $\langle L^{\cosh - 1}, L\log(L+1)\rangle$ by the pair of Zygmund spaces $\langle L_{exp}, L\log L\rangle$.

\begin{proposition}
\label{jeden}
Let $(Y, \Sigma, \mu)$ be a $\sigma$-finite measure space and $L\log(L+1)$ be the Orlicz space defined by the Young's function $x \mapsto x\log(x+1), \ x\geq0$. Then $L\log(L+1)$ is an equivalent renorming of the K\"othe dual of $L^{\cosh -1}$.
\end{proposition}
\begin{proof}
Firstly observe that there are $0<a \leq b< \infty$ such that

$$ \phi_1(ax) \leq \phi_2(x) \leq \phi_1(bx)$$
for $x>0$, where $\phi_1(x) = \cosh x - 1$ and $\phi_2(x) = e^x - x -1$. Consequently, $\phi_1\approx \phi_2$ (even globally equivalent, cf \cite{RR}, Section 2.2). Hence, the conjugate function of $\cosh x - 1$ is equivalent to the conjugate function of $e^x - x -1$, namely $(x+1)\log(x+1) - x$.
Finally, observe that there are $0<c \leq d < \infty$ such that 
$$\psi_1(cx) \leq \psi_2(x) \leq \psi_1(bx)$$
for $x>0$, where $\psi_1(x) = (x+1)\log(x+1) - x$ and $\psi_2(x) = x\log(x+1)$. Thus $\psi_1\approx \psi_2$, and the proof is complete.
\end{proof}

We wish to close this Section with an analysis of the relation between the pair of Orlicz spaces $\langle L^{\cosh -1}, L \log (L+1)\rangle$ and the pair of Zygmund spaces $\langle L_{exp}, L \log L\rangle$ for the finite measure case.

\begin{proposition}
\label{dwa}
For finite measure spaces $(\X, \Sigma, m)$ one has
\begin{equation}
L^{\cosh - 1} = L_{exp}.
\end{equation}
Consequently, for the finite measure case, $\langle L^{\cosh - 1}, L\log(L+1)\rangle$ is an equivalent renorming of $\langle L_{exp}, L\log L\rangle $.
\end{proposition}
\begin{proof}
As $\Phi(t)= e^t - t -1 \approx \cosh t -1$, see Proposition \ref{jeden}, it is enough to prove that $L^{\Phi} = L_{exp}$. To show this, firstly recall
that the Young's function $\Phi_{exp}(t)$ defining $L_{exp}$ is equal to 
\begin{equation}
\Phi_{exp}(t) = \left \{
\begin{array}{ll}
t & \quad \text{if \quad $0 \leq t \leq 1$}\\
\frac{1}{e} e^t &  \quad \text{if \quad $t >1$.}\\
\end{array} \right. 
\end{equation}
Secondly, 
\begin{equation}
\lim_{t \to \infty} \frac{e^t - t - 1}{\Phi_{exp}(t)} = \lim_{t \to \infty}\frac{e^t -t -1}{\frac{1}{e} e^t} = e.
\end{equation}
Hence there exists $u_0>0$ and some $K>1$ so that
\begin{equation}
\frac{1}{K} (e^t -t -1) \leq \Phi_{exp}(t) \leq K (e^t - t -1)
\end{equation}
for $t\geq u_0$. Given a function $f$ we therefore have
\begin{equation}
\int\Phi_{exp}(|f|) \chi_{E}(x) d m(x) < \infty \Longleftrightarrow \int (e^{|f|} -|f| -1) \chi_E(x) dm(x) < \infty
\end{equation}
where $E= \{ x \in X: |f(x)| \geq u_0 \}$. Next, let $M_0$, $M_1$ respectively be the maximal value of $\Phi_{exp}$ and $e^t -t -1$ on $[0, u_0]$.
Then
$$\int \Phi_{exp}(|f|) \chi_{E^c}(x) dm(x) \leq M_0  \int dm(x) < \infty$$
and
$$\int(e^{|f|} - |f| - 1) \chi_{E^c}(x) dm(x) \leq M_1 \int dm(x) < \infty,$$
where $E^c \equiv X \setminus E$.
If we combine this with the earlier observation, then we get that
\begin{equation}
\int \Phi_{exp}(|f|) dm < \infty \Longleftrightarrow \int(e^{|f|} - |f| -1) dm < \infty,
\end{equation}
which proves the claim.
\end{proof}

\section{Classical regular systems \cite{LM}}
We begin with the definition of the classical regular model (cf \cite{PS}). Let $\{\Omega, \Sigma, \nu\}$
be a probability space; $\nu$ will be called the reference measure. The set of  densities of all the probability measures equivalent to $\nu$ will be called the state space $\mathcal{S}_{\nu}$, i.e.
\begin{equation}
{\mathcal{S}}_{\nu} = \{ f  \in L^1(\nu): f>0 \quad \nu-a.s.,\, E(f) =1 \},
\end{equation}
where, $E(f) \equiv \int f  d\nu$. It is worth pointing out that $f \in \mathcal{S}_{\nu}$ implies that $fd\nu$ is a probability measure.
\begin{definition}
\label{clasmodel}
The classical statistical model consists of the measure space $\{\Omega, \Sigma, \nu\}$, state space $\mathcal{S}_{\nu}$, and the set of measurable functions $L^0(\Omega, \Sigma, \nu)$.
\end{definition}
To select regular random variables, i.e.  random variables having all finite moments,
we define the moment generating functions as follows: fix $f \in {\mathcal{S}}_{\nu}$, take a real random variable $u$ on $(\Omega, \Sigma, f d\nu)$
and define:
\begin{equation}
\hat{u}_f(t) = \int exp(tu) f d\nu, \qquad t \in \mathbb{R}. 
\end{equation}

Note that
$t \mapsto \hat{u}_f(t)$ is called the Laplace transform of $u$ cf \cite{B}.
 In the sequel we will need the following properties of $\hat{u}$ ({for details see Widder, \cite{Wid}}):

\begin{enumerate}
\item $\hat{u}$ is analytic in the interior of its domain,
\item its derivatives are obtained by differentiating under the integral sign.
\end{enumerate}
Now the following definition is clear (cf \cite{PS}):
\begin{definition}
The set of all random variables on $(\Omega, \Sigma, \nu)$ such that for a fixed $f \in S_{\nu}$
\begin{enumerate}
\item{} $\hat{u}_f$ is well defined in a neighborhood of the origin $0$,
\item{} the expectation of $u$ is zero,
\end{enumerate}
will be denoted by $L_f \equiv L_f(f \cdot \nu)$ and called the set of regular random variables.
\end{definition}

 The set of  regular random variables having zero expectation is characterized by:

\begin{theorem}(Pistone-Sempi, \cite{PS})
$L_f$ is the closed subspace of the Orlicz space $L^{\cosh - 1}(f\cdot \nu)$ of zero expectation random
variables.
\end{theorem}

Consequently, the first space in the postulated pair, see (\ref{KwantOrlicz}),  has appeared as the natural home for regular observables.
But as $L\log(L+1)$ is the K\"othe dual of $L^{\cosh -1}$, see Proposition \ref{jeden}, the appearance of the second Orlicz space in (\ref{KwantOrlicz}) is also explained. In particular, elements in $L\log(L+1)$ can be considered as ``normal'' functionals over the space $L^{\cosh - 1}$ of regular observables. 

Turning to the entropy problem, we note (see Remark \ref{uwaga1}(2)) that there is a relation $\succ$  between the Young's function $x\log(x + \sqrt{1+x^2}) - \sqrt{1 + x^2} +1$ and the entropic function $c \cdot x\log x$ where $c$ is a positive number. Consequently, as the Orlicz space defined by the Young's function 
$x\log(x + \sqrt{1+x^2}) - \sqrt{1 + x^2} +1$ is equal to $L \log(L+1)$ (cf Proposition \ref{jeden}),
the condition $f \in L\log(L+1)$ guarantees that the continuous
entropy is well defined for the finite measure case. Thus we arrived at:
\begin{corollary}
$$\langle L^{\cosh -1}, L\log(L+1)\rangle$$
or equivalently
$$\langle L_{exp}, L\log L\rangle$$
provides the proper framework for the description of classical regular statistical systems (based on probability measures).
\end{corollary}
\begin{proof}
Note that regular statistical systems are reliant on finite measures $f \cdot \nu$, so the claim is a direct consequence of the previous Section (cf Proposition \ref{dwa} and Corollary \ref{dwa1}).
\end{proof}
An analysis of the classical continuous entropy for the infinite measure case will be presented in Section 4 below.

\section{Applications of Orlicz space technique to Boltzmann's theory}

The goal of this section is twofold. Firstly, we want to present another example illustrating how the Orlicz space technique is useful in Statistical Mechanics. Secondly, we have studied the continuous entropy $H(f) = - \int f(x)\log f(x) dx$ only for the finite measure case. Now, we wish to show that if $H(f)$ is defined in the context of the Orlicz space $L\log(L+1)$ (or $L\log L$), the natural 
Lyapunov functional for Boltzmann's equation, namely $H_+(f) \equiv - H(f)$, is then well defined.

Recall that (spatially homogeneous) Boltzmann's equation reads:
\begin{equation}
\frac{\partial f_1}{\partial t} = \int d\Omega \int d^3v_2 I(g, \theta)|{\bf v}_2 - {\bf v}_1| (f^{\prime}_1 f^{\prime}_2 - f_1f_2)
\end{equation}
where $f_1 \equiv f({\bf v}_1,t), f_2^{\prime} \equiv f({\bf v}^{\prime}_2,t)$, etc, are velocity distribution functions, with $\bf v$ standing for velocities before collision, and ${\bf v}^{\prime}$ for velocities after collision.
$I(g, \theta)$ denotes the differential scattering cross section, $d\Omega$ is the solid angle element, and $g = |{\bf v}|$. As it was mentioned, the natural Lyapunov functional for this equation is 
the continuous entropy with opposite sign, i.e. $H_+(f) = \int f(x) \log f(x) dx$, where $f$ is supposed to be a solution of Boltzmann's equation. The exceptional status of the functional $H_+(f)$ in an analysis of Boltzmann's equation follows from McKean's result \cite{Mckean}. He proved that the entropy $H(f)$ is the only increasing functional for some simplified model of gas. The time behaviour of $H_+(f)$ is described by the H-Theorem (see \cite{Thom} for physical aspects of Boltzmann's equation while a survey of the mathematical theory of this equation can be found in \cite{Vil}; see also \cite{CC}, \cite{DP1}, and \cite{DP2}).

One of the features of H-functional $H_+(f)$ where $f \in L^1$ is the fact that it is unbounded both from below and from the above (see eg \cite{Gar}). We wish to show that Orlicz space technique allows a rigorous analysis of H-functional so also H-Theorem.  We start with

\begin{proposition} Let $f \in L^1\cap L\log L$ where both Orlicz spaces are over $(\Rn^3, \Sigma, d^3v)$ ($d^3v$ - the Lebesgue measure). Then
\begin{equation}
\label{L3}
\int|f|\log |f| d^3v = \lim_{\epsilon\searrow 0} \int_{E^f_{\epsilon}} |f| \log |f| d^3v
\end{equation}
is well defined, and bounded above. Moreover, each $\int_{E^f_{\epsilon}} |f| \log |f| d^3v$ is finite, where 
$E^f_{\epsilon} = \{v: |f(v)| > \epsilon \}$.
\end{proposition}

\begin{proof}
Let $f \in L^1$. For any $\epsilon >0$ we have
\begin{equation}
|f| \chi_{(\epsilon, \infty)}(|f|) \geq\epsilon \chi_{(\epsilon, \infty)}(|f|).
\end{equation}
Since $\int|f|\chi_{(\epsilon, \infty)}(|f|)d^3v \leq \|f\|_1 < \infty$, the set $E^f_{\epsilon}$ must have finite measure. Note that further
\begin{equation}
\int_{E^f_{\epsilon}} |f|\log |f| d^3v \geq - e^{-1} \int_{E^f_{\epsilon}} d^3v > - \infty.
\end{equation}
Moreover, for $f \in L^1 \cap L \log L$,
\begin{equation}
\infty > \int|f| \log^+ |f| d^3v = \int_{E^f_1}|f| \log^+ |f| d^3v + \int_{(E^f_1)^c} |f| \log^+ |f| d^3v 
\end{equation}
$$> \int_{E^f_{\epsilon}}|f| \log^+ |f| d^3v > - \infty.$$
Thus
$$\int_{E^f_{\epsilon}}|f| \log^+ |f| d^3v$$
is well defined.
Clearly, (\ref{L3}) holds and the proof is complete.
\end{proof}

Consequently, using the Orlicz space technique, the continuous entropy $H(f)$ of any velocity distribution function $f \in L^1\cap L\log L$ can be uniformly approximated by distributions (states) 
with well defined continuous entropy.

As in this Section we are concerned with the infinite measure case, we have that $L\log L \neq L\log(L+1)$. Hence the following result is relevant.
\begin{proposition}
\label{oho}
Let $f \in L^1\cap L\log(L+1)$, $f\geq 0$, where both Orlicz spaces are over $(\Rn^3, \Sigma, d^3v)$ ($d^3v$ - the Lebesgue measure). Then
\begin{equation}
\label{L4}
H_{\epsilon}(f)= \int f \log (f+ \epsilon) d^3v 
\end{equation}
is well defined for any $\epsilon >0$
\end{proposition}
\begin{proof}
Let $f \in L^1\cap L\log(L+1)$ with $f\geq 0$. As both spaces are vector spaces, then $\beta f \in L^1\cap L\log(L+1)$ for an arbitrary $\beta > 0$. It is an exercise to see the Young's functions $t$ and $t\log(t+1)$ both satisfy the $\Delta_2$ condition globally. Hence we even have that $\int \beta f \log(\beta f + 1) d^3v <\infty$. Note that
\begin{equation}
\label{L5}
\int  \beta f \log(\beta f + 1) d^3v = \int (\beta \log \beta) f d^3v + \beta \int f \log(f + \frac{1}{\beta}) d^3v.
\end{equation}
As the LHS of (\ref{L5}) and the first term of the RHS of (\ref{L5}) are finite numbers, the claim follows.
\end{proof}

Let us comment on the above results.
\begin{enumerate}
\item Proposition \ref{oho} implies that for any $f\in  L^1\cap L\log(L+1)$, $f\geq 0$, $H_+(f)$ (so also $H(f)$) can be approximated by finite numbers $H_{\epsilon}(f)$.
\item The important point to note here is the fact that DiPerna-Lions (see \cite{DP1}, \cite{DP2}, and \cite{AV})
showed that the estimates
\begin{equation}
f \in L^{\infty}_t([0,T]; L^1_{x,v}((1 +|v|^2 +|x|^2)dxdv) \cap L\log (L +1))
\end{equation}
and
\begin{equation}
D(f) \in L^1([0,T] \times \Rn^N_x),
\end{equation}
where $D(f) = \frac{1}{4} \int d\Omega \int d^3v_1 d^3v_2 I(g, \theta)|{\bf v}_2 - {\bf v}_1| (f^{\prime}_1 f^{\prime}_2 - f_1f_2) \log\frac{f^{\prime}_1 f^{\prime}_2}{f_1f_2}$, 
are sufficient to build a mathematical theory of weak solutions.

Furthermore, Villani announced, see \cite{Vil}, Chapter 2, Theorem 9, that for particular cross sections (collision kernels in Villani's terminology) weak solutions of Boltzmann equation are in $L\log(L+1)$.
\item Consequently, for the infinite measure case, the condition $f \in L^1 \cap L \log(L+1)$ is well suited to entropic problems associated to Boltzmann's equation.
\end{enumerate}

As the entropic functional $H_+(f)$ plays so important a role in the analysis of Boltzmann's theory, we will continue the examination of its properties.

\begin{proposition}\label{prop4.3}
Let $f \in L^1 \cap L \log(L+1)$ and $f\geq 0$. Then
$$H_+(f) = \int f \log f d^3v$$
is bounded above, and if in addition $f\in L^{1/2}$ (equivalently $f^{1/2}\in L^1$), it is also bounded from below. Thus $H_+(f)$ is bounded below on a dense subset of the positive cone of $L \log(L+1)$.
\end{proposition}
\begin{proof}
As $x \mapsto \log x$, $x>0$ is a monotonic function, we have $x\log x \leq x\log(x+ \epsilon)$, for any $\epsilon >0$.
Hence
$$H_+(f) \leq H_{\epsilon}(f).$$
To examine boundedness from below, let $0\leq f \in L^1 \cap L\log(L+1)$ be such that $f^{\frac{1}{2}} \in L^1$. Then
denoting $\int f^{\frac{1}{2}}d^3v$ by $N$, one has
\begin{equation}
\int f \log f d^3v = \int f^{\frac{1}{2}} \cdot f^{\frac{1}{2}} \log(f^{\frac{1}{2}} \cdot f^{\frac{1}{2}}) d^3v
=2 \int(f^{\frac{1}{2}} \log(f^{\frac{1}{2}})) \cdot f^{\frac{1}{2}}d^3v \end{equation}
$$\geq 2 (\int f d^3v) \log(\frac{\int f d^3v}{N}),$$
where the last inequality follows from Jensen's inequality. (See \cite{rud} for a very general version of this inequality.) 

The last part of the claim will be established if we can show that each nonnegative element of $L^1 \cap L \log(L+1)$ is the norm limit of a sequence of functions with support having finite measure. We present a very general proof of this fact which can be directly translated to the noncommutative context. Let $f\in L^1 \cap L \log(L+1)$ be given with $f\geq 0$, and for any $n\in \mathbb{N}$ let $E_n=\{v | \tfrac{1}{n}\leq f(v)\leq n\}$. By \cite[Corollary 3.3]{L} the sequence $\{f\chi_{E_n}\}$ will in fact converge to $f$ in the $L^1$-norm. Next notice that the Young's function $\Psi(t) = t\log(t+1)$ generating $L \log(L+1)$ is actually an $N$-function. (That means that the limit formulae $\lim_{t\to 0}\frac{\Psi(t)}{t}=0$ and $\lim_{t\to \infty}\frac{\Psi(t)}{t}=\infty$ are valid.) It is an exercise to see that $\Psi(t)$ also satisfies the $\Delta_2$-condition globally. Hence by \cite[Remark 6.8]{LM}, $L\log(L+1)$ must have order-continuous norm. (The remark referred to assumes that $N$-functions are in view.) Since $f\chi_{E_n}$ increases pointwise to $f$ as $n\to\infty$, the order continuity of the norm ensures that $f\chi_{E_n}$ converges to $f$ in the $L\log(L+1)$-norm as $n\to\infty$. Thus $f\chi_{E_n}$ converges to $f$ in the norm on $L^1 \cap L\log(L+1)$.
\end{proof}

To sum up, \textit{ the proposed approach is compatible with a rigorous analysis of Boltzmann's equation}.

In the last two Sections we have shown that the scheme for classical statistical mechanics based on the two distinguished 
Orlicz spaces $\langle L^{\cosh -1}, L\log(L+1)\rangle $ does work. However, the basic theory for Nature is Quantum Mechanics. Therefore the question of a quantization of the given approach must be considered. This will be done in the next two Sections. To facilitate the procedure of quantization, although up to now only classical systems have been considered, we have deliberately tried to formulate our arguments in as general a way as possible.

\section{Non-commutative Orlicz spaces}\label{s5}
For reader's convenience we start this section by presenting a brief review of quantum (noncommutative) Orlicz spaces extracted from Section 2 in \cite{LM}.

Let $\Phi$ be a given Young's function. In the context of semifinite von Neumann algebras 
$\cM$ equipped with an fns (faithful normal semifinite) trace $\tau$, the space of all $\tau$-measurable operators 
$\widetilde{\cM}$ (equipped with the topology of convergence in measure) plays the role of 
$L^0$ (for details see \cite{Ne}). In this  case, Kunze \cite{Kun} 
used this identification to define the associated noncommutative Orlicz space to be 
$$L^{ncO}_{\Phi}{} = \cup_{n=1}^\infty n\{f \in \widetilde{\cM} : \tau(\Phi(|f|) \leq 1\}$$ 
and showed that this is a linear space which becomes a Banach space when equipped with the 
Luxemburg-Nakano norm $$\|f\|_\Phi = \inf\{\lambda > 0 : \tau(\Phi(|f|/\lambda)) \leq 
1\}.$$Using the linearity it is not hard to see that $$L^{ncO}_{\Phi}{} = \{f \in 
\widetilde{\cM} : \tau(\Phi(\lambda|f|)) < \infty  \quad \mbox{for some} \quad \lambda = 
\lambda(f) > 0\}.$$ 
Thus there is a clear analogy with the commutative case.

It is worth pointing out that
there is another approach to Quantum Orlicz spaces. Namely, one can replace $(\cM, \tau)$ by $(\cM, \varphi)$, where $\varphi$ is a normal faithful state on $\cM$ (for details  see \cite{AlRZ}).
However, as we wish to put some emphasis on the universality of quantization,
we prefer to follow the Banach space theory approach developed by Dodds, Dodds and de Pagter \cite{DDdP} .

Given an element $f \in \widetilde{\cM}$ and $t \in [0, \infty)$, the generalized singular 
value $\mu_t(f)$ is defined by $\mu_t(f) = \inf\{s \geq 0 : \tau(\I - e_s(|f|)) \leq t\}$ 
where $e_s(|f|)$ $s \in \mathbb{R}$ is the spectral resolution of $|f|$. The function $t \to 
\mu_t(f)$ will generally be denoted by $\mu(f)$. For details on the generalized singular values 
see \cite{FK}.  (This directly extends classical notions where for any $f \in L^0_{\infty}{}$, 
the function $(0, \infty) \to [0, \infty] : t \to \mu_t(f)$ is known as the decreasing 
rearrangement of $f$.) We proceed to briefly review the concept of a Banach Function Space of 
measurable functions on $(0, \infty)$. (Necessary background is given in \cite{DDdP}.) A function norm 
$\rho$ on $L^0(0, \infty)$ is defined to be a mapping $\rho : L^0_+ \to [0, \infty]$ satisfying
\begin{itemize}
\item $\rho(f) = 0$ iff $f = 0$ a.e.  
\item $\rho(\lambda f) = \lambda\rho(f)$ for all $f \in L^0_+, \lambda > 0$.
\item $\rho(f + g) \leq \rho(f) + \rho(g)$ for all .
\item $f \leq g$ implies $\rho(f) \leq \rho(g)$ for all $f, g \in L^0_+$.
\end{itemize}
Such a $\rho$ may be extended to all of $L^0$ by setting $\rho(f) = \rho(|f|)$, in which case 
we may then define $L^{\rho}(0, \infty) = \{f \in L^0(0, \infty) : \rho(f) < \infty\}$. If 
now $L^{\rho}(0, \infty)$ turns out to be a Banach space when equipped with the norm 
$\rho(\cdot)$, we refer to it as a Banach Function space. If $\rho(f) \leq \lim\inf_n(f_n)$ 
whenever $(f_n) \subset L^0$ converges almost everywhere to $f \in L^0$, we say that $\rho$ 
has the Fatou Property. If less generally this implication only holds for $(f_n) \cup \{f\} 
\subset L^{\rho}$, we say that $\rho$ is lower semi-continuous. If further the situation $f 
\in L^\rho$, $g \in L^0$ and $\mu_t(f) = \mu_t(g)$ for all $t > 0$, forces $g \in L^\rho$ and 
$\rho(g) = \rho(f)$, we call $L^{\rho}$ rearrangement invariant (or symmetric). Using the 
above context Dodds, Dodds and de Pagter \cite{DDdP} formally defined the noncommutative space 
$L^\rho(\widetilde{\cM})$ to be $$L^\rho(\widetilde{\cM}) = \{f \in \widetilde{\cM} : \mu(f) \in 
L^{\rho}(0, \infty)\}$$and showed that if $\rho$ is lower semicontinuous and $L^{\rho}(0, 
\infty)$ rearrangement-invariant, $L^\rho(\widetilde{\cM})$ is a Banach space when equipped 
with the norm $\|f\|_\rho = \rho(\mu(f))$. 

Now for any Young's function $\Phi$, the Orlicz 
space $L^\Phi(0, \infty)$ is known to be a rearrangement invariant Banach Function space 
with the norm having the Fatou Property, see Theorem 8.9 in \cite{BS}. Thus on selecting $\rho$ to be 
 $\|\cdot\|_\Phi$, the very general framework of Dodds, 
Dodds and de Pagter presents us with an alternative approach to realising noncommutative 
Orlicz spaces.

As the von Neumann entropy is defined on $\M=B(H)$ we end this Section with a description of the Banach Function spaces for $B(H)$ 
which are constructed using the philosophy of Dodds, Dodds and de Pagter described above.
Let $\M=B(H)$ equipped with the standard trace $\mathrm{Tr}$. Then  $\tM = B(H)$ \cite{Te1}. Let $n$ be a non-negative integer and let $b\in B(H)$ be given. Since $\mathrm{Tr}$ is integer-valued on the projection lattice of $B(H)$, it follows from \cite[Proposition 2.4]{FK} that $\mu_t(b)=\mu_n(b)=a_{n+1}(b)$ for any $t\in [n,n+1)$, where $a_{n+1}$ is the distance from $b$ to the operators with rank 
at most $n$ (the so-called $(n+1)$-th approximation number of $b$ \cite{Pie}). Of course $b$ will be compact if and only if $a_n(b)\to 0$ as $n\to \infty$. If indeed $b$ is compact, then by a result of Allahverdiev (cf. \cite[Theorem II.2.1]{GK}), the $a_n(b)$'s correspond 
to the elements of the spectrum of $|b|$ arranged in decreasing order according to multiplicity. Given a Banach Function norm $\rho$, the 
prescription given above (cf. \cite{DDdP}) says that $b\in L^\rho(B(H))$ if and only if $\mu(b)\in L^\rho(0,\infty)$, with the norm on $L^\rho(B(H))$ given by $\|b\|_\rho=\|\mu(b)\|_\rho$. Now let $\Phi$ be a Young's function. Then $b\in L^\Phi(B(H))$ if and only if $\mu(b)\in L^\Phi(0,\infty)$ if and only if there exists some $\alpha>0$ so that 
$\int_0^\infty \Phi(\alpha\mu_t(b))\,dt = \sum_{n=0}^\infty\Phi(\alpha(a_n(b))<\infty$ if and only if $\{a_n(b)\}$ belongs to the Orlicz sequence space $\ell_\Phi(\mathbb{N})$. Similarly the Luxemburg norm of $b\in L^\Phi(B(H))$ can then be shown to be precisely $\|b\|_\Phi= \inf\{\epsilon>0:\sum_{n=0}^\infty\Phi(a_n(b)/\epsilon) \leq 1\}$ (the Luxemburg norm of $\{a_n(b)\}$ considered as an element of $\ell_\Phi(\mathbb{N})$.

\section{Non-commutative regular systems}

In \cite{LM} the definition of non-commutative regular system was given. To quote this method of quantization we need some preparation (cf \cite{LM}). Let
 $({\M}, \tau)$ be a pair consisting of a semifinite von Neumann algebra and fns trace.
 
 \begin{remark}
 \begin{enumerate}
 \item For large quantum systems, i.e. for systems with an infinite number of degrees of freedom, type III factors are of paramount interest. Namely (cf \cite{Haag})
 representations of quasilocal algebras induced by an equilibrium state as well as local algebras of relativistic theory in the so called the vacuum sector, lead to type III factors. However by using crossed-product techniques (cf \cite{Te1}), one arrives at semifinite algebras.
 \item Since in the models of Quantum Physics von Neumann algebras act on separable Hilbert spaces, one can restrict oneself to $\sigma$-finite algebras (cf \cite{BR}).
 The advantage of this assumption follows from the fact that it allows for a simplification of the crossed-product technique in that here one can more easily select the elements of the original algebra (cf \cite{wata})
 \end{enumerate}
 \end{remark}
 
 Consequently, {\it the assumption that $\M$ is a semifinite algebra acting on a separable Hilbert space is not too restrictive} as one can consider instead of $\cN$ (factor III) the corresponding crossed product $\M= \cN \rtimes_{\sigma}\Rn$ with a nice identification of $\cN$ inside $\M$.

 To provide the promised preliminaries let us define (see \cite{Tak}, vol. I):
\begin{enumerate}
\item $n_{\tau} = \{ x \in {\M}: \tau(x^*x) < + \infty \}.$
\item ({\it definition ideal of the trace} $\tau$) $m_{\tau} = \{xy: x,y \in n_{\tau} \}.$
\item $\omega_x(y) = \tau(xy), \quad x\geq 0.$
\end{enumerate}

One has (for details see Takesaki, \cite{Tak}, vol. I)
\begin{enumerate}
\item if $x \in m_{\tau}$, and $x\geq 0$, then $\omega_x \in \M_*^+$. 
\item If $L^1(\M,\tau)$ stands for the completion of $(m_{\tau}, ||\cdot||_1)$ then $L^1(\M, \tau)$ is isometrically isomorphic to $\M_*$.
\item $\M_{*,0} \equiv \{ \omega_x : x \in m_{\tau} \}$ is norm dense in $\M_*$.
\end{enumerate}
Finally, denote by $\M_*^{+,1}$ ($\M_{*,0}^{+,1}$) the set of all normalized normal positive functionals in $\M_*$ (in $\M_{*,0}$ respectively).
Now, performing a ``quantization'' of Definition \ref{clasmodel} we arrive at (cf \cite{LM})
\begin{definition}
\label{qsm}
The noncommutative statistical model consists of a quantum measure space $(\M, \tau)$, ``quantum densities with 
respect to $\tau$'' in the form of $\M_{*,0}^{+,1}$, and
the set of $\tau$-measurable operators $\tM$.
\end{definition}

 Having ``quantized'' the statistical model, we can present the definition of regular noncommutative statistical model (\cite{LM}).
\begin{definition}
\label{kwant}
\begin{equation}
L^{quant}_x = \{ g \in \tM: \quad 0 \in D(\widehat{\mu_{x}^g(t)})^0, \quad x \in m_{\tau}^+ \},
\end{equation}
where $D(\cdot)^0$ stands for the interior of the domain $D(\cdot)$ and
\begin{equation} 
\widehat{\mu_{x}^g(t)} = \int \exp(t\mu_s(g)) \mu_s(x)ds, \qquad t\in\mathbb{R}.
\end{equation}
({\it Notice that the requirement that $0 \in D(\widehat{\mu_{x}^g(t)})^0$, presupposes that the transform $\widehat{\mu_{x}^g(t)}$ is well-defined in a neighborhood of the origin.})
\end{definition}

We remind that above and in the sequel $\mu(g)$ ($\mu(x)$)  stands for the function $[0,\infty) \ni t \mapsto \mu_t(g) \in [0, \infty]$ 
($[0,\infty) \ni t \mapsto \mu_t(x) \in [0, \infty]$ respectively).

To give a non-commutative generalization of the Pistone-Sempi theorem, we need a generalization of the Dodds, Dodds, de Pagter approach, i.e. the approach which was presented in Section 4. To this end we need \cite{LM}

\begin{definition}
\label{def5.4}
Let $x \in L_+^1(\M, \tau)$ and let $\rho$ be a Banach function norm on $L^0((0, \infty), \mu_t(x)dt)$. In the spirit of \cite{DDdP} we then formally define the weighted noncommutative Banach function space $L^{\rho}_{x}(\tM)$ to be the collection of all $f \in \tM$ for which $\mu(f)$ belongs to $L^{\rho}((0, \infty), \mu_t(x)dt)$. For any such $f$ we write $\|f\|_\rho = \rho(\mu(f))$.
\end{definition}

\begin{remark}
The classical statistical model is constructed using objects of the form $f \cdot d\nu$. A faithful noncommutative translation of this would be to look at objects of the form $\tau(x^{\frac{1}{2}} \cdot x^{\frac{1}{2}}) = \int \mu_t(x^{\frac{1}{2}}\cdot x^{\frac{1}{2}}) dt$. However it is convenient for us to rather use the related objects 
$\tau_x(\cdot) = \int \mu_t(\cdot) \mu_t(x)dt$. These two objects are clearly closely related, with $\tau_x$ having the advantage of exhibiting many trace-like properties.
\end{remark}
The mentioned generalization of the Dodds, Dodds, de Pagter approach is contained in:
\begin{theorem}\cite{LM}
\label{ncbf}
Let $x \in L_+^1(\M, \tau)$. Let $\rho$ be a rearrangement-invariant Banach function norm on $L^0((0, \infty), \mu_t(x)dt)$ which satisfies the Fatou property and such that:
$\nu(E) < \infty \Rightarrow \rho(\chi) < \infty$ and $\nu (E) < \infty \Rightarrow \int_E f d\nu \leq C_E \rho(f)$
for some positive constant $C_E$, depending on $E$ and $\rho$ but independent of $f$ ($\nu$ stands for $\mu_t(x)dt$).
 Then $L^{\rho}_{x}(\tM)$ is a linear space and $\|\cdot\|_\rho$ a norm. Equipped with the norm $\|\cdot\|_\rho$, $L^{\rho}_{x}(\tM)$ is a Banach space which injects continuously into $\tM$.
\end{theorem}

and the generalization of the Pistone-Sempi is given by  ( see \cite{LM})

\begin{theorem}\label{QPS}
The set $L^{quant}_x$  coincides with the the weighted Orlicz space $L_x^{\cosh - 1}(\tM) \equiv L^{\Psi}_{x}(\tM)$ (where $\Psi = \cosh -1$) of noncommutative regular random variables.
\end{theorem}

Finally, to show that statistics and thermodynamics can be well established for noncommutative regular statistical systems, we note that 
for elements $x\in L^1_+(\mathcal{M})$, $\mu_t(x)dt$ gives a finite resonant measure on $(0, \infty)$. To see this, it is enough to observe that $t \mapsto \mu_t(\cdot)$
is a nonincreasing and right continuous function (see \cite{FK}). Note, this property of $\mu_t(\cdot)dt$ simplifies the theory of rearrangement-invariant Banach function spaces. In particular, one can easily apply the scheme given in Section 2.
Moreover, both of the spaces $L{\log(L+1)}(\tM)$ and $L{\log L}(\tM)$ are suitable frameworks within which to study the quantum entropy $\tau(f\log(f))$. We justify this claim by first proving a quantum version of Propositions \ref{oho} and \ref{prop4.3}.

\begin{proposition}\label{prop6.8}
Let $\cM$ be a semifinite von Neumann algebra with an fns trace $\tau$ (cf Section 4) and let $f \in L^1\cap L\log(L+1)(\tM)$, 
$f\geq 0$. Then
$\tau(f\log(f + \epsilon))$ is well defined for any $\epsilon >0$. Moreover
$$\tau(f \log f)$$
is bounded above, and if in addition $f\in L^{1/2}$ (equivalently $f^{1/2}\in L^1$), it is also bounded from below. Thus $\tau(f \log f)$ 
is bounded below on a dense subset of the positive cone of $L \log(L+1)$.
\end{proposition}

\begin{proof}
Let $f \in L^1\cap L\log(L+1)(\tM)$ be given with $f\geq 0$. From the discussion in Section \ref{s5}, we know that this forces $\mu(f) \in L^1\cap L\log(L+1)(0,\infty)$. Notice that a similar argument to the one used in the proof of Proposition \ref{oho}, can now be used to show that $\int_0^\infty|\mu_t(f)\log(\mu_t(f)+\epsilon)|\,dt<\infty$ for any $\epsilon$. Since $t\to \mu_t(f)$ is non-increasing, the fact that $\mu(f) \in L^1(0,\infty)$, ensures that $\mu_t(f)$ decreases to zero as $t\to\infty$. Hence we also have that $\tau(g(f)) = \int_0^\infty g(\mu_t(f))\,dt$ for any non-negative Borel function $g$ with $g(0)=0$ (see \cite[Remark 3.3]{FK}). If we combine this with the above observation regarding $\mu(f)$, it follows that for any $\epsilon>0$ we have 
$$\tau(|f\log(f+\epsilon)|)=\int_0^\infty|\mu_t(f)\log(\mu_t(f)+\epsilon)|\,dt<\infty.$$This proves the first claim. 

To prove the second claim, fix some $\epsilon>0$. Using the fact that $x \to \log(x)$ ($x>0$) is monotonic, we may conclude from the Borel functional calculus that $f\log(f) \leq f\log(f+\epsilon)$. Let $\chi_I$ denote the spectral projection of $f$ corresponding to the interval $I$. Since 
$\log$ is non-negative on $[1,\infty)$, it therefore follows from the above inequality that 
$0\leq f\chi_{[1,\infty)}\log(f\chi_{[1,\infty)}) \leq f\chi_{[1,\infty)}\log(f\chi_{[1,\infty)}+\epsilon)\leq |f\log(f+\epsilon)|$. Hence $0\leq \tau(f\chi_{[1,\infty)}\log(f\chi_{[1,\infty)}))\leq \tau(|f\log(f+\epsilon)|)<\infty$. Notice that $0\geq f\chi_{[0,1)}\log(f\chi_{[0,1)})$. We may therefore give meaning to $\tau(f\chi_{[0,1)}\log(f\chi_{[0,1)}))$ by setting $\tau(f\chi_{[0,1)}\log(f\chi_{[0,1)}))= -\tau(-f\chi_{[0,1)}\log(f\chi_{[0,1)}))$ and to $\tau(f\log(f))$ by setting $\tau(f\log(f))=\tau(f\chi_{[0,1)}\log(f\chi_{[0,1)})+ \tau(f\chi_{[1,\infty)}\log(f\chi_{[1,\infty)}))$. Then $\tau(f\log(f))$ is well defined (possibly assuming the value $-\infty$), and bounded above by $\tau(|f\log(f+\epsilon)|)$.

It remains to prove the final claim. To this end let $0\leq f \in L^1 \cap L\log(L+1)(\tM)$ be such that $f^{\frac{1}{2}} \in L^1(\tM)$. We have already observed that then $\mu(f)\in L^1 \cap L\log(L+1)(0,\infty)$. Since $\mu(f)^{1/2}=\mu(f^{1/2})$ (see \cite[Lemma 2.5]{FK}), the assumption regarding $f^{1/2}$ similarly ensures that $\mu(f)^{1/2}\in L^1(0,\infty)$. A similar argument to the one used in the proof of Proposition \ref{prop4.3} now ensures that in this case 
$$-\infty < \int_0^\infty\mu_t(f)\log(\mu_t(f))\,dt<\infty.$$
Hence by \cite[Remark 3.3]{FK} we then have that $$\tau(|f\log(f)|)=\int_0^\infty|\mu_t(f)\log(\mu_t(f))|\,dt<\infty.$$This in turn ensures that $\tau(f\log(f))>-\infty$.
\end{proof}

One also has the following ``almost'' characterisation of the elements of $L{\log L}(\tM)^+$ for which $f\log(f)$ is integrable. 

\begin{proposition}
As before let $\cM$ be a semifinite von Neumann algebra with an fns trace $\tau$. Let $f = f^* \in \widetilde{\mathcal{M}}$ be given. 
By $\chi_I$ will denote spectral projection of $f$ corresponding to the interval $I$.
If $f\in L{\log L}(\tM)^+$ with $\tau(\chi_{[0,1]})<\infty$, then $\tau(|f\log(f)|)$ exists (i.e. $f\log(f)\in L^1(\tM)$).

Conversely if $\tau(|f\log(f)|)$ exists , then $f\in L{\log L}(\tM)^+$ with $\tau(\chi_I)<\infty$ for any open subinterval $I$ of $[0,1]$.
\end{proposition}

Before proving this Proposition, we discuss the significance of the condition $\tau(\chi_{[0,1]})<\infty$. For any $f\geq 0$, membership of $L{\log L}(\tM)$ ensures that $f\log(f)\chi_{[1,\infty)} \in L^1(\tM))$. (Here $\chi_{[1,\infty)}$ is a spectral projection of $f$.) This follows from the fact that the Young's function generating this space is $t\log^+(t)<\infty$. However to be sure that in fact $f\log(f)\in L^1(\tM)$, we need some additional criteria with which to control the portion $f\log(f)\chi_{[0,1]}$. The requirement that $\tau(\chi_{[0,1]})<\infty$, is precisely such a criterion. Consequently, if the ``state'' is taken from the noncommutative Zygmund space $L_{\log L}(\widetilde{M})$ and $\tau( \chi_{[0,1]})< \infty$, then the entropy function exists!

\begin{proof}
Firstly we show: if $f\in L{\log L}(\widetilde{M})^+$ with $\tau(\chi_{[0,1]})<\infty$, then $\tau(|f\log(f)|)$ exists.
To this end note that
$f\in L{\log L}(\widetilde{M})^+$  guarantees that \\ $0\leq\tau(f\chi_{[1,\infty)}\log(f\chi_{[1,\infty)}))< \infty$ since $\log(f\chi_{[1,\infty)})=\log^+(f\chi_{[1,\infty)})$. Now notice that $0\geq f\chi_{[0,1]}\log(f\chi_{[0,1]})\geq - \frac{1}{e}\chi_{[0,1]}$. So then 
$\tau(|f\chi_{[0,1]}\log(f\chi_{[0,1]})|) \leq \frac{1}{e}\tau(\chi_{[0,1]})<\infty$. Hence $\tau(|f\log(f)|)<\infty$.

\bigskip
Conversely, we show that if $\tau(|f\log(f)|)$ exists, then $f\in L{\log L}(\tM)^+$ and 
for any $0<\delta < \frac{1}{e}<\epsilon<1$ we have that $\tau(\chi_{[\delta,\epsilon]}(f))<\infty$.

Notice that $0\leq \tau(|f\chi_{[1,\infty)}\log(f\chi_{[1,\infty)})|)\leq \tau(|f\log(f)|)$. Furthermore, one has
$f\chi_{[1,\infty)}\log(f\chi_{[1,\infty)})=f\log^+(f)$, which means that the above inequality  ensures that 
$f\in L{\log L}(\widetilde{M})^+$. For the final part of the claim note that $t\log(t)$ is negative valued on $[0,1]$, 
decreasing on $[0,e^{-1})$ and increasing on $(e^{-1},1]$. These facts ensure that 
\begin{eqnarray*}
0 &\geq& \delta\log(\delta)\chi_{[\delta, 1/e]} + \epsilon\log(\epsilon)\chi_{[1/e,\epsilon]}\\
&\geq& f\chi_{[\delta, 1/e]}\log(f\chi_{[\delta, 1/e]}) + f\chi_{[1/e,\epsilon]}\log(f\chi_{[1/e,\epsilon]})\\
&=& f\chi_{[\delta,\epsilon]}\log(f\chi_{[\delta,\epsilon]}).
\end{eqnarray*}
So for $0<K<\{|\delta\log(\delta)|, |\epsilon\log(\epsilon)|\}$ we will have 
$$K\tau(\chi_{[\delta,\epsilon]}) \leq \tau(|f\chi_{[\delta,\epsilon]}\log(|f\chi_{[\delta,\epsilon]})|) \leq \tau(|f\log(f)|)<\infty.$$
This completes the proof of the Proposition.
\end{proof}

\begin{remark}
We briefly consider the significance of the above Proposition for more general settings.
\begin{enumerate}
\item  Thus far we have studied the entropy function $x \mapsto x\log x$ for semifinite algebras. However it is non-semifinite type III $W^*$-algebras that seem to be the rule for infinite systems. The crossed-product technique provides a tool for bridging this gap, in that such type III algebras can be represented as subalgebras of semifinite algebras. Moreover using this technique, the noncommutative Orlicz spaces corresponding to such type III algebras can the be constructed using semifinite algebras. 
\item We briefly describe how the quantum Orlicz spaces mentioned above may be constructed for a $\sigma$-finite von Neumann algebra $\M$ with an fns state $\phi$ by means of the crossed-product technique. For such algebras one has (cf \cite{wata})
\begin{equation}
L^1(\cM) =  \mathrm{closure}(h^{\frac{1}{2}}_{\phi} \cM h^{\frac{1}{2}}_{\phi})
\cong \cM_*
\end{equation}
where $h^{\frac{1}{2}}_{\phi}$ is an unbounded operator (equal to the Radon-Nikodym derivative of the extension $\tilde{\phi}$ of $\phi$ on the crossed-product $\A = \cM \rtimes_{\sigma} \Rn$ with respect to the canonical trace $\tau_{\A}$ on $\A$). 

Starting from $L^1(\M)$, we can now define the required quantum Orlicz spaces. Let $\Psi$, $\Phi$ be a pair of the complementary Young's functions. To define the quantum Orlicz space $L^\Psi(\M)$, we first make use of the norm of the K\"othe dual (namely $L^\Phi(0,\infty)$) of $L^\Psi(0,\infty)$ to define the function 
$$\theta_\Phi(t)=|\!|\!|\chi_{[0,t]}|\!|\!|_\Phi\quad t\geq 0.$$(The precise form of the norm on $L^\Phi(0,\infty)$ will depend on the norm we start with on $L^\Psi(0,\infty)$.) The function $\theta_\Phi$ is the so-called fundamental function of $L^\Phi(0,\infty)$ (cf. \cite{BS}). In the case of $L^p(0,\infty)$ spaces, the associated fundamental function is just $\theta_p(t)=t^{1/p}$. Using $\theta_\Phi$, we now define the quantum Orlicz space $L^\Psi(\cM)$ to be the space of all (possibly unbounded) operators $f$ in $\widetilde{\A}$ for which 
$\widetilde{\theta}_\Phi(h)^{1/2}f\widetilde{\theta}_\Phi(h)^{1/2}\in L^1(\cM)$. For such spaces the quantity 
$\mu_1(f)$ turns out to be a quasi-norm in terms of which all convergence properties can be described. If we apply this construction to a semifinite algebra $\M$ equipped with an fns trace $\tau_{\M}$, we end up with a space which is an exact copy of the space $L^\Psi(\tM)$ 
produced using the techniques described in Section \ref{s5}. Details of the above construction may be found in \cite{LL}.
\end{enumerate}
\end{remark}

In conclusion, analogous to the commutative case, we get the following conclusion.

\begin{corollary} Either of the pairs 
$$\langle L^{\cosh -1}, Llog(L+1)\rangle$$
or 
$$\langle L_{exp}, L\log L\rangle$$
provides an elegant rigorous framework for the description of non-commutative regular statistical systems,
where now the Orlicz (and Zygmund) spaces are noncommutative.
\end{corollary}

We wish to close this Section with an examination of von Neumann entropy $S(\varrho) = - \mathrm{Tr} (\varrho \log \varrho)$, where $\varrho$ is a density matrix 
(on a Hilbert space $H$) and $\mathrm{Tr}$ is the canonical trace on $B(H)$  (so the entropy is considered on $B(H)$). In this case $\widetilde{B(H)} = B(H)$, 
i.e. from noncommutative measure theory the von Neumann algebra $B(H)$ presents an exceptional case (see the last paragraph of Section 5). We already mentioned that from a physical point of view, type I von Neumann algebras are not well suited for the description of infinite quantum systems. Nevertheless, the von Neumann quantum entropy plays so important a role in the description of simple systems (see \cite{Weh}, \cite{OP}) that the proposed examination is justified.

The final paragraphs of Section 5 and the proposed approach imply that the set of regularized states will be given by $L\log(L+1)(B(H))$ with particularly nice behaviour as far as entropy is concerned exhibited by those states which also belong to $L^1(B(H))$, where $L^1(B(H))$ stands for the trace-class operators (we remind that trace class operators form the predual of $B(H)$ and that $L^1 \approx \cM_*$). Let $0\leq \varrho \in L^1$ be given. Hence $\varrho = \sum \lambda_i P_{x_i}$, $\lambda_i\geq 0$, $\sum\lambda_i < \infty$, where $P_{x_i}$ is an orthogonal projector onto the unit vector $x_i \in H$, and where $\{x_i\}$ forms an orthonormal system in $H$. We may additionally assume that the $\lambda_i$'s are arranged in decreasing order. But then we must have that $\lambda_i$ decreases to 0 (or else $\sum\lambda_i < \infty$ will fail). Since $\log$ is increasing on $[1, \infty)$, this in turn ensures that
\begin{equation}
\label{A1} 
0\leq \lambda_i\log(\lambda_i+1) \leq K\lambda_i\quad\mbox{for all}\quad i
\end{equation}
where $K=\log(\lambda_1+1)$. But then $$\sum \lambda_i 
\log(\lambda_i + 1) < \infty,$$which by the discussion in Section 5 ensures that $\varrho \in L\log(L+1)$. Thus in this exceptional case one gets that $L^1 \subset L\log(L+1)$. Since for any $\alpha \in [0,1]$ we have that $\alpha \leq \alpha^{1/2}$, it similarly follows that $L^{1/2} \subset L^1$ in this case. 

Repeating the argument given in Section 4, one gets: {\it $S(\varrho)$ is bounded from below on $L^1$, and from above on the subspace $L^{1/2}$. Thus if on the basis of Propositions \ref{prop4.3} and \ref{prop6.8} one prefers the space $L^1 \cap L\log(L+1)$ to $L\log(L+1)$ in our approach, then for this very exceptional case that will yield the pair
\begin{equation}
\langle B(H), L^1(B(H))\rangle .
\end{equation}
Thus the approach presented in this paper, canonically extends the elementary quantum theory based on the above pair.}

To elucidate the peculiarity of the considered case we note
\begin{enumerate}
\item The Banach function space $\ell_{\Phi}(\mathbb{N})$ is defined on infinite, completely atomic measure space (with all atoms having equal measure). Therefore the inclusions given by (2.10) are valid for $B(H)$.
\item Observe that here 
\begin{equation}
\label{koncowauwaga}
L^1 \approx L^1(B(H)) \subset Llog(L+1) \subset B(H) \approx L^{\infty}
\end{equation}
which is completely opposite to (2.11).
\item The considered quantization of simple models leads to (\ref{koncowauwaga})
\item The above argument is not valid for large systems described by factors of type III and II.
\item For a nonatomic measure space, as considered in Section 4, inclusions of type (\ref{koncowauwaga}) are not true.
\end{enumerate}
Finally we note that the arguments given in Section 4 and (\ref{A1}),  
imply that the functionals $S_{\epsilon}(\varrho) = - \mathrm{Tr} (\varrho \log ( \varrho + \epsilon)$, $\epsilon >0$, provide well defined approximations of the von Neumann entropy $S(\varrho)$.

\section{Acknowledgments}

The support of the grant number N N202 208238 as well as   the Foundation for Polish Science TEAM
project cofinanced by the EU European Regional Development Fund
for W. A. Majewski and a grant from the National Research Foundation for L. E. Labuschagne is gratefully acknowledged. Any opinion, findings and conclusions or recommendations expressed in this material, are those of the authors, and therefore the NRF do not accept any liability in regard thereto.

\end{document}